\documentclass[11pt,a4paper]{article}

\usepackage{amsmath}
\usepackage{amsthm}
\usepackage{makeidx}
\usepackage{latexsym}
\usepackage{graphicx}
\usepackage{amssymb}
\usepackage{hyperref}
\usepackage{float}

\usepackage[usenames,dvipsnames]{pstricks}
\usepackage{epsfig}

\newtheorem{theorem}{Theorem}
\newtheorem{example}{Example}

\newtheorem{prop}{Proposition}

\newtheorem{definition}{Definition}
\theoremstyle{remark}
\newtheorem{remark}{Remark}

\date{}
\usepackage[english]{babel}
\usepackage[latin1]{inputenc}

\newcommand{\group}{{\mathcal G}}
\newcommand{\pell}{{\mathcal H}}
\newcommand{\conic}{{\mathcal C}}

\newcommand{\para}{{\mathcal P}}
\newcommand{\ec}{{\mathcal E}}

\newcommand{\ZZ}{{\mathbb Z}}
\newcommand{\KK}{{\mathbb K}}
\newcommand{\trap}{{\tau}}
\newcommand{\iso}{{\Phi}}
\DeclareMathOperator{\lcm}{lcm}
\providecommand{\keywords}[1]{\textbf{\textbf{Keywords:}} #1}

\makeindex

\begin{document}

\title{An efficient and secure RSA--like cryptosystem exploiting R\'{e}dei rational functions over conics}
\author{Emanuele Bellini, Nadir Murru}

\maketitle

\begin{abstract}
We define an isomorphism between
the group of points of a conic and
the set of integers modulo a prime equipped with a non-standard product.
This product can be efficiently evaluated through the use of R\'edei rational functions.
We then exploit the isomorphism
to construct a novel RSA-like scheme.
We compare our scheme with classic RSA and with RSA-like schemes based on the cubic or conic equation.
The decryption operation of the proposed scheme turns to be two times faster than RSA,
and involves the lowest number of modular inversions with respect to other RSA-like schemes based on curves.
Our solution offers the same security as RSA in a one-to-one communication and more security in broadcast applications.
\end{abstract}

\section{Introduction}
RSA cryptosystem security is based on the existence of an one-way trapdoor function \cite{goldwasser1996lecture}, \cite{yao1982theory}, which is easy to compute and difficult to invert without knowing some information, i.e., the trapdoor.
Consider $N=pq$, where $p$ and $q$ are two primes of roughly the same size,
and $e$ an invertible element in $\mathbb{Z}_{\phi(N)}$ ($\phi(N)$ Euler totient function). 
Given the function $\trap_{N,e}(x) = x^e \pmod N$,
it is not known if there exists a probabilistic polynomial time algorithm in $N$
which can invert $\trap$, for any $x\in \mathbb{Z}_N^*$,
without knowing either $p,q,\phi(N)$ or the inverse $d$ of $e$ in $\mathbb{Z}_{\phi(N)}$.
Thus the pair $(N,e)$, called the public key, is known to anyone, 
while the triple $(p,q,d)$, called the secret key, 
is only known to the receiver of an encrypted message.
The ciphertext $C$ is obtained as $C=M^e \pmod N$, and
the original message $M$ is obtained with the exponentiation $M = C^d \pmod N$.

Factoring $N$ is the most common tried way for breaking RSA.
Some attacks are possible when 
either the private exponent $d$ is small \cite{wiener1990cryptanalysis}, 
or the public exponent $e$ is small \cite{coppersmith1996low}, \cite{coppersmith1997small}.
Several other methods, which exploit extra information leaked by erroneous implementations of the cryptosystem, exist (e.g., see \cite{lenstra2012ron}).

Beside one-to-one communication scenario, 
there also exist other cryptographic scenarios, such as broadcast applications, 
where RSA leaks additional vulnerabilities (e.g., see \cite{hastad1986n}).

RSA-like schemes (e.g., see \cite{koyama1992new}, \cite{koyama1995fast}, \cite{padhye2006public}) 
have been proposed in order to avoid some of these attacks. Some of these schemes turn to have also a faster decryption procedure. 
They are based on isomorphism between two groups, 
one of which is the set of points over a curve, usually a cubic or a conic.
%

In this work we present an improvement of such schemes. 
We provide a new RSA-like scheme based on isomorphism over a conic, 
with the fastest isomorphism with respect to the ones known by the authors.
Furthermore, our scheme uses a different set of conics 
with respect to that of \cite{padhye2006public}, 
which owns the fastest decryption procedure.
We provide security proofs and efficiency comparisons.
%

In Section \ref{sec:related} we give an overview of the use of conic equations in cryptography and 
we introduce RSA-like schemes based on isomorphism.
In Section \ref{sec:conic} we introduce a parametrization of certain conics and we introduce the isomorphism used for the cryptographic scheme.
In Section \ref{sec:scheme} we describe our scheme, highlighting connections with R\'{e}dei rational functions and differences with other schemes based on the Pell equation and Dickson polynomials.
In Section \ref{sec:security} we discuss security issues.
We prove that our scheme is as secure as RSA 
in an one-to-one communication scenario, 
while it is more secure in a broadcast scenario. 
We also show that our scheme is secure against partially known-plaintext attacks.
Finally, in Section \ref{sec:efficiency} we expose some efficiency considerations and comparisons 
with special focus on the decryption procedure, 
where our scheme presents the lowest computational complexity
if compared to all known similar schemes.
\section{Related works}\label{sec:related}
In this section we provide a quick overview of the use of conic equations to construct cryptographic protocols.
Then we define a generic RSA-like scheme based on an isomorphism between two groups.
\subsection{The use of conic equations in cryptography}\label{sec:pell_in_crypto}
The use of conic equations is not novel in cryptography, in particular to build RSA-like cryptosystems.\\
In \cite{lemmermeyer2006introduction} the Pell analogous of 
RSA protocol, Diffie-Hellman key-exchange, and computing square roots modulo $n$,
is presented. 
Specifically,
as far as it concerns RSA, 
it is noticed that where RSA sends one encrypted
message $C$ about the size of $N$, the Pell version has to send twice as many bits per message, 
without having increased the security of the system.\\
In \cite{gysin1999use} RSA, Rabin and ElGamal variants of Pell's equation are presented proving that the proposed solutions are as secure as the original schemes. 
The proposed schemes are claimed to have also the same asymptotic complexity as the original schemes.\\
%
Three RSA-like schemes based on Pell's equation are presented in \cite{padhye2006public}. 
All three schemes have a faster decryption procedure than RSA.
The first two schemes are proved to be as secure as RSA in a one-to-one communication scenario and more secure in a broadcast scenario since they increase the complexity of some low exponent attacks such as \cite{hastad1986n}.
The third scheme is the Pell analogous of \cite{sakurai2002new} scheme. 
It is proved to be semantically secure (recall RSA is not)
and it is derived by randomizing the second scheme in a standard way described in \cite{sakurai2002new}.\\
An implementation of Scheme II of \cite{padhye2006public}
is analyzed in \cite{sarma2011public}. 
The scheme is implemented using GMP library \cite{gmp505}, and compared in its favor against a RSA implementation using the same library \cite{biswas2003fast}.\\
The use of R\'edei rational function for a cryptographic scheme is introduced in \cite{kameswari2012cryptosystems}. Though, the two schemes are inefficients with respect to RSA, since a larger modulo ($\mathcal{O}(n^2)$ instead of $\mathcal{O}(n)$) is used for the public exponent, and no isomorphism is exploited, making both encryption and decryption slower than RSA.\\
In \cite{segar2013pell} the advantages of key generation of a cryptosystem based on Pell's equation versus standard RSA (and some variants) key generation are discussed. The authors claim their key generation prevents Wiener attacks \cite{wiener1990cryptanalysis}.\\
A symmetric cryptographic scheme based on the Brahmagupta-Bh\~{a}skara equation
has been proposed in \cite{murthy2006cryptographic},
attacked in \cite{youssef2007comment}, corrected in the author's replay,
and re-attacked in \cite{alvarez2008known} with a known-plaintext attack.\\
A combination of Arnold Cat Map, Chaotic Map, 
and Brahmagupta-Bh\~{a}skara equation has also been used 
to build cryptosystems for image encryption and decryption 
(e.g., see \cite{rao2011vlsi}, \cite{thomasvlsi}).\\
Based on Pell's equation, in \cite{raoidentity}, a claimed Identity-Based Encryption scheme is presented. 
The efficiency of its implementation using GMP library \cite{gmp505} is analyzed in \cite{rao2015public}.\\
Also a Dynamic Threshold Multi Secret Sharing scheme \cite{rao2013model} 
and a Signature scheme \cite{mishra2014signpell} have been recently published.\\
Other RSA type cryptosystems based on elliptic curves exist, such as
\cite{koyama1995fast},
\cite{koyama1992new}, and
\cite{demytko1994new}. 
These schemes have a computationally more expansive addition operation compared to schemes based on Pell's equation, such as \cite{padhye2006public}.
%
%
%
%
%
%
%
%
%

\subsection{RSA-like schemes based on isomorphism}\label{sec:rsa_like}
In the following we use multiplicative notation for groups. 
In particular, for a group $\group$ with operation $\odot_{\group}$ and an element $G \in \group$, 
we use the following notation for the exponentiation with respect to the group operation
$$G^{\odot_{\group} x} = \underbrace{G \odot_{\group} \ldots \odot_{\group} G}_{x \text{ times }} \,.$$
In \cite{koyama1995fast}, the concept of RSA-type schemes based on the isomorphism between two groups 
has been introduced  in order to have a faster decryption with respect to the original RSA scheme.
Such a scheme can be obtained by finding an efficiently computable isomorphism 
$$\iso:(\group,\odot_{\group}) \to (\group',\odot_{\group'})$$ 
between two groups, 
in such a way that, in the decryption procedure,
the use of the exponentiation in  $\group'$ 
and the inverse of the isomorphism is more efficient than just applying the exponentiation in $\group$. 
The scheme is visualized in Figure \ref{fig:rsa_iso}.
The message $M$ is encrypted into the ciphertext $C'$, either
by first applying the exponentiation in $\group$ and then the isomorphism, or by first applying the isomorphism and then the exponentiation in $\group'$. 
The decryption follows the inverse process.\\
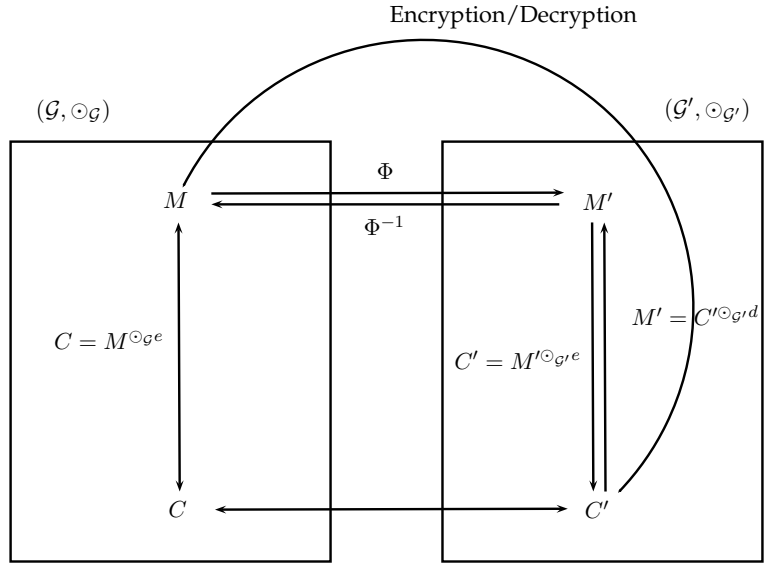
\begin{figure}
\scalebox{0.8} 
{
\begin{pspicture}(0,-4.232957)(14.569062,5.02048)
\psframe[linewidth=0.04,dimen=outer](6.23,2.7670429)(0.93,-4.232957)
\usefont{T1}{ppl}{m}{n}
\rput(3.6645312,1.7770429){$M$}
\usefont{T1}{ppl}{m}{n}
\rput(3.6945312,-3.342957){$C$}
\psline[linewidth=0.04cm,arrowsize=0.05291667cm 2.0,arrowlength=1.4,arrowinset=0.4]{<->}(3.73,-3.052957)(3.71,1.4070429)
\psframe[linewidth=0.04,dimen=outer](13.33,2.7670429)(8.03,-4.232957)
\usefont{T1}{ppl}{m}{n}
\rput(10.604531,1.7570429){$M'$}
\usefont{T1}{ppl}{m}{n}
\rput(10.594531,-3.342957){$C'$}
\psline[linewidth=0.04cm,arrowsize=0.05291667cm 2.0,arrowlength=1.4,arrowinset=0.4]{<-}(10.71,1.4070429)(10.73,-3.052957)
\psline[linewidth=0.04cm,arrowsize=0.05291667cm 2.0,arrowlength=1.4,arrowinset=0.4]{<-}(4.25,1.7070429)(9.97,1.7070429)
\psline[linewidth=0.04cm,arrowsize=0.05291667cm 2.0,arrowlength=1.4,arrowinset=0.4]{<-}(10.049961,1.9012433)(4.25,1.8870429)
\usefont{T1}{ppl}{m}{n}
\rput(7.1045313,2.257043){$\iso$}
\usefont{T1}{ppl}{m}{n}
\rput(7.094531,1.3570429){$\iso^{-1}$}
\usefont{T1}{ppl}{m}{n}
\rput(2.5645313,-0.5429571){$C=M^{\odot_{\group} e}$}
\usefont{T1}{ppl}{m}{n}
\rput(12.204531,-0.12295712){$M'=C'^{\odot_{\group'} d}$}
\usefont{T1}{ppl}{m}{n}
\rput(1.9845313,3.257043){$(\group,\odot_{\group})$}
\usefont{T1}{ppl}{m}{n}
\rput(12.4045315,3.2770429){$(\group',\odot_{\group'})$}
\rput{-116.00077}(11.132247,6.9560957){\psarc[linewidth=0.04,arrowsize=0.05291667cm .0,arrowlength=1.4,arrowinset=0.4]{<->}(7.739416,0.0){4.43269}{72.086815}{269.3873}}
\usefont{T1}{ppl}{m}{n}
\rput(9.211094,4.817043){Encryption/Decryption}
\psline[linewidth=0.04cm,arrowsize=0.05291667cm 2.0,arrowlength=1.4,arrowinset=0.4]{<-}(10.53,-3.052957)(10.51,1.4070429)
\usefont{T1}{ppl}{m}{n}
\rput(9.284532,-0.8629571){$C'=M'^{\odot_{\group'} e}$}
\psline[linewidth=0.04cm,arrowsize=0.05291667cm 2.0,arrowlength=1.4,arrowinset=0.4]{<->}(10.1099615,-3.3387568)(4.31,-3.352957)
\end{pspicture} 
}
\caption{RSA-like scheme based on isomorphism.}
\end{figure}\label{fig:rsa_iso}
Usually $\group$ is the set of points of a curve over the ring $\ZZ_N$  provided with a group law, $N$ the product of two primes, and $\group'$ is a subgroup of $\ZZ_N^*$.
The first proposed curve in \cite{koyama1995fast} was a singular cubic curve, while a particular set of conics was proposed in \cite{padhye2006public} providing a more efficient isomorphism inverse. In our work we propose another set of conics and we consider $\group'=\ZZ_N^*$ provided with a non-standard group law, yielding a more efficient isomorphism inverse.
%
%
\section{Product of points over conics} \label{sec:conic}
In this section, 
we use an irreducible polynomial of degree 2 to define a quotient field 
that induces a product over certain conics. 
This product gives a group structure to the conics. 
Then, we introduce a parametrization that allows us 
to define a bijection between a conic and a set of parameters,
yielding a group structure on this set. 
Finally, we determine the order of this group.
In Section \ref{sec:scheme} we use the bijection, which turns out to be an isomorphism,
to construct our RSA-like scheme.
\subsection{A group structure over conic curves}
Conics are the most simple non-linear curves and the Pell hyperbola defined by
$$\pell=\{(x,y)\in\mathbb R\times\mathbb R: x^2-Dy^2=1\},$$
where $D$ is a given positive integer (non-square), is one of the most popular, since it contains all the solutions of the famous Pell equation. A group structure can be defined over the Pell hyperbola by means of the following product between points in $\pell$:
$$(x_1,y_1)\cdot(x_2,y_2)=(x_1x_2+Dy_1y_2,y_1x_2+x_1y_2),$$
see, e.g., \cite{Veb} and \cite{Jac}.

Let $\mathbb K$ be an ordinary field, the previous product can be generalized in order to provide a group structure over the conics
$$\conic_\KK=\{(x,y)\in \mathbb K\times \mathbb K: x^2+Hxy-Dy^2=1\}.$$
When clear from the context we write $\conic$ instead of $\conic_\KK$, omitting the subscript.\\
Let $x^2-Hx-D$ be an irreducible polynomial over $\mathbb K[x]$ and let us consider the quotient field
$$\mathbb A=\mathbb K[x]/(x^2-Hx-D).$$
Given $p+q x$ and $r+s x$ in $\mathbb A$, the induced product is
$$(p+q x)(r+s x)=(pr+qs D)+(qr+ps+qs H)x\,.$$
We now want to define the conjugate over $\mathbb A$ 
of an element $p + q x$ with respect to $\mathbb K$.
Notice that $\mathbb A$ is an extension of degree 2 over $\mathbb K$
and so the minimal polynomial of $p+qx$ 
over $\mathbb K$ has degree 2.
Thus, we are looking for an element $r+sx$ such that
the sum and the product between this element and its conjugate $p+qx$
is in $\mathbb K$.
It is immediate to see that $s=-q$.
Moreover, since
$$(p+qx)(r-qx)=(pr-q^2D)+(rq-pq-q^2 H)x$$
we obtain that $r=p+qH$.
Thus, the conjugate $\overline{p+qx}$ of an element $p+qx$ is $(p+qH)-qx$.
Then, the norm of $p+qx$ over $\mathbb K$
is given by the product of the conjugates of $p+qx$, which is
$$(p+qx)(\overline{p+qx})=p^2+Hpq-Dq^2 \,.$$
The group of unitary elements of $\mathbb A^*=\mathbb A -\{0\}$ is
$$\mathcal U=\{p+q x\in\mathbb A^*:p^2+Hpq-Dq^2=1\}.$$
Consequently, there is a bijection between $\mathcal U$ and the conic $\conic$ which determines a commutative group structure over the conic by means of the product
\begin{equation}\label{prodE} (x,y)\odot_\conic(u,v)=(xu+yvD, yu+xv+yvH),\quad \forall (x,y), (u,v)\in \conic.  \end{equation}
\begin{prop}
$(\conic,\odot_\conic)$ is a commutative group with identity $(1,0)$ and the inverse of an element $(x,y)$ is $(x+Hy,-y)$.
\end{prop}
\subsection{A group structure on the set of parameters}
We can use the following parametrization for the conic $\conic$:
$$y=\cfrac{1}{m}(x+1).$$
\begin{definition}\label{def:iso_k}
Let us define the set of parameters $\para_\KK=\KK \cup \{\alpha\}$, with $\alpha$ not in $\mathbb K$. From the parametrization we can derive the following bijection between $\conic$ and $\para_\KK$:
\begin{equation} 
\iso_{H,D}:
\begin{cases} 
\conic & \rightarrow \para_\KK \cr 
(x,y)  & \mapsto \cfrac{1+x}{y} \quad \forall (x,y)\in \conic, \quad y\not=0 \cr 
(1,0)  & \mapsto\alpha \cr
(-1,0) & \mapsto-\cfrac{H}{2} \ , 
\end{cases} 
\end{equation}
and
\begin{equation}\label{para} 
\iso^{-1}_{H,D}:
\begin{cases}
\para_\KK  & \rightarrow \conic \cr 
m      & \mapsto \left(\cfrac{m^2+D}{m^2+Hm-D}\ , \cfrac{2m+H}{m^2+Hm-D}\right)\quad \forall m \in \mathbb K \cr 
\alpha & \mapsto(1,0)\ , 
\end{cases}. 
\end{equation}
\end{definition}

For the seek of simplicity, we only write $\iso$ when there is no confusion.
This bijection allows us to derive from $\odot_\conic$ the following product over the set of parameters $\para_\KK$:
\begin{equation}
\label{prodP} 
\begin{cases} 
a\odot_{\para_\KK} b=\cfrac{D+ab}{H+a+b}, \quad & a+b\not=-H \cr 
a\odot_{\para_\KK} b=\alpha,              \quad & a+b=-H 
\end{cases}.  
\end{equation}

\begin{prop}
$(\para_\KK,\odot_{\para_\KK})$ is a commutative group with identity $\alpha$ and the inverse of an element $a$ is the element $b$ such that $a+b=-H$ and $\iso$ is an isomorphism between $\conic$ and $\para_\KK$.
\end{prop}

The above parametrization of the conic $\conic$ has been introduced in \cite{bcm}, where the authors used it in order to study approximations of irrationalities over conics.
The commutative group $(\para_\KK,\odot_{\para_\KK})$ can be also directly derived from $\mathbb A^*$ and it is possible to show that the quotient group $\mathcal B=\mathbb A^*/\mathbb K^*$ and $\para_\KK$ are isomorphic \cite{bcm}.\\
Let us observe that when $\mathbb K = \mathbb Z_p$, with $p$ prime, 
$\mathbb A$ is a finite field with $p^2$ elements, i.e., it is the Galois field $GF(p^2)$.
Moreover, $\mathcal B$ has order $(p^2-1)/(p-1)=p+1$ and consequently we have that $\para_{\ZZ_p}$ and $\conic_{\ZZ_p}$ are cyclic groups of order $p+1$. 
\begin{remark} \label{remark:fermat}
Thus, an analogous of the Fermat little theorem holds in $\para_{\ZZ_p}$ and $\conic_{\ZZ_p}$, i.e.,
$$m^{\odot_{\para_{\ZZ_p}} (p+2)}=m \pmod p,\quad \forall m\in \para_{\ZZ_p}$$
and
$$(x,y)^{\odot_{\conic_{\ZZ_p}} (p+2)}=(x,y) \pmod p,\quad \forall (x,y)\in \conic_{\ZZ_p},$$
where powers are performed using products $\odot_{\para_{\ZZ_p}}$ and $\odot_{\conic_{\ZZ_p}}$, respectively.
\end{remark}

\section{A public-key cryptosystem}\label{sec:scheme}
In this section we develop a RSA-like scheme 
exploiting the properties of the product $\odot_{\para_{\ZZ_p}}$ and the isomorphism $\iso$
redefined to make sense also when the conic is considered over the ring $\ZZ_N$. 
In particular, we will exploit the fact that 
an analogous of the Fermat little theorem holds in $\para_{\ZZ_p}$. 
Moreover, we see that powers with respect to this product can be evaluated in a fast way and 
the decryption operation in our scheme involves only one modular inverse, 
by means of the definition of $\iso$. 

\subsection{Preliminaries}
In order to insert a trapdoor in our scheme we would like to extend the ideas of Section \ref{sec:conic} to hold for the conic
$$\pell_{\ZZ_N}=\{(x,y)\in \mathbb Z_N\times \mathbb Z_N: x^2-Dy^2=1\pmod N\}\,.$$
where $\pm D\in\mathbb Z_N^*$ are quadratic non-residues modulo $N$ and $N=pq$ with $p,q$ prime numbers.
The condition on $D$ will ensure the isomorphism inverse to be well defined.
The condition on $-D$ will ensure the powers with respect to $\odot_{\para_{\ZZ_N}}$ to be well defined. 
To lighten the notation, in the remainder of the paper we omit the subscript $\ZZ_N$ from $\pell_{\ZZ_N}$ and $\para_{\ZZ_N}$ where there is no ambiguity.

Thus, we are working on a conic $\conic$ where $H=0$ and $\mathbb K=\mathbb Z_N$. Since $\ZZ_N$ is not a field we need to refine some of the definitions. 
In particular, in this case, 
the map $\iso_{0,D}:\pell\to\para$ of Definition \ref{def:iso_k} is not an isomorphism. 
However, considering
$$\pell^*=\{(x,y)\in \mathbb Z_N\times \mathbb Z_N^*: x^2-Dy^2=1\pmod N\}\,$$
we have that $\lvert \pell^* \rvert=\lvert \mathbb Z_N^* \rvert$ as proved in the following proposition. From now on we also omit the subscripts ${0,D}$ from $\iso_{0,D}$.
\begin{prop}
With the above notation, we have that
\begin{enumerate}
\item $\forall (x_1,y_1), (x_2,y_2)\in \pell^*$, $\iso(x_1,y_1)=\iso(x_2,y_2)\Leftrightarrow (x_1,y_1)=(x_2,y_2)$;
\item $\forall m_1, m_2\in \mathbb Z_N^*$, $\iso^{-1}(m_1)=\iso^{-1}(m_2)\Leftrightarrow m_1=m_2$;
\item $\forall m\in \mathbb Z_N^*$, we have $\iso^{-1}(m)\in \pell^*$ and $\forall (x,y)\in\pell^*$, we have $\iso(x,y)\in\mathbb Z_N^*$.
\end{enumerate}
\end{prop}
\begin{proof}
\begin{enumerate}
\item We have that 
$$\iso(x_1,y_1)=\iso(x_2,y_2)\Leftrightarrow (1+x_1)y_2=(1+x_2)y_1.$$
Squaring both members and multiplying them by $D$, we obtain
$$D(1+x_1)^2y_2^2=D(1+x_2)^2y_1^2.$$
Since $(x_1,y_1), (x_2,y_2)\in \pell^*$, we get
$$(1+x_1)^2(x_2^2-1)=(1+x_2)^2(x_1^2-1)^2$$
from which it is easy to prove that $x_1=x_2$ and $y_1=y_2$.
\item We have that
$$\iso^{-1}(m_1)=\iso^{-1}(m_2)\Leftrightarrow \left( \cfrac{m_1^2+D}{m_1^2-D},\cfrac{2m_1}{m_1^2-D}\right)=\left(\cfrac{m_2^2+D}{m^2-D}, \cfrac{2m_2}{m_2^2-D} \right).$$
Direct calculations show that this equality holds only when $m_1=m_2$.
\item Using the explicit forms of $\iso$ and $\iso^{-1}$ the proof is straightforward.
\end{enumerate}
\end{proof}
We recall the definition of $\iso_{H,D}$ with $H=0$ where we consider it only between $\pell^*$ and $\mathbb Z_N^*$.
\begin{definition}
We consider the bijection $\iso$ between $\pell^*$ and $\mathbb Z_N^*$ as 
$$
\iso:
\begin{cases} 
\pell^* & \rightarrow \mathbb Z_N^* \cr 
(x,y)   & \mapsto \cfrac{1+x}{y}  
\end{cases}
$$
and its inverse
$$ 
\iso^{-1}:
\begin{cases}
\mathbb Z_N^*  & \rightarrow \pell^* \cr 
m      & \mapsto \left(\cfrac{m^2+D}{m^2-D}\ , \cfrac{2m}{m^2-D}\right)
\end{cases}.
$$
\end{definition}

The products over $\pell^*$ and $\mathbb Z_N^*$ still follow the same rules seen in the previous section, i.e.,
$$(x_1,y_1)\odot_{\pell}(x_2,y_2)=(x_1x_2+Dy_1y_2,y_1x_2+x_1y_2), \quad \forall (x_1,y_1), (x_2,y_2)\in \pell \,.$$
and
$$a\odot_\para b=\cfrac{D+ab}{a+b},\quad \forall a,b\in \mathbb Z_N^*,\quad a+b\in \mathbb Z_N^*. $$

As a consequence of Remark \ref{remark:fermat}, given $l = \lcm(p+1,q+1)$ and $r = 1 \pmod l$ we have
$$m^{\odot_\para r}=m \pmod N,\quad \forall m\in \mathbb Z_N^*$$
and
$$(x,y)^{\odot_{\pell} r}=(x,y) \pmod N,\quad \forall (x,y)\in \pell_N^*,$$
where powers are performed using products $\odot_\para$ and $\odot_{\pell}$, respectively.
%
\begin{remark}\label{rem:redei_exp_complexity}
In \cite{bcm2} the authors prove that powers with respect to the product $\odot_\para$
can be evaluated by the R\'{e}dei rational functions \cite{Redei}. 
These functions arise from the development of $(z+\sqrt{D})^n$,
where $z$ is an integer and $D$ is a non-square positive integer. 
One can write
\begin{equation*}
\label{pow}(z+\sqrt{D})^n=A_n(D,z)+B_n(D,z)\sqrt{D},
\end{equation*}
where
$$ A_n(D,z)=\sum_{k=0}^{[n/2]}\binom{n}{2k}D^kz^{n-2k}, \quad B_n(D,z)=\sum_{k=0}^{[n/2]}\binom{n}{2k+1}d^kz^{n-2k-1}.$$
These polynomials can be also determined by
$$M^n=\begin{pmatrix} A_n(D,z) & DB_n(d,z) \cr B_n(D,z) & A_n(D,z) \end{pmatrix}$$
with 
$$M=\begin{pmatrix} z & D \cr 1 & z \end{pmatrix}.$$
The R\'{e}dei rational functions $Q_n(D,z)$ are defined by 
$$Q_n(D,z)=\cfrac{A_n(D,z)}{B_n(D,z)}, \quad \forall n\geq1.$$
The powers with respect to the product $\odot_\para$ can be evaluated as follows:
$$z^{\odot_\para n}=\underbrace{z\odot_\para ... \odot_\para z}_{\text{$n$ times}}=Q_n(D,z).$$
In this way, the exponentiation in $\para$ can be performed efficiently. 
Indeed, in \cite{More} the author exhibits an algorithm  of complexity $O(log_2(n))$ with respect to addition, subtraction and multiplication to evaluate R\'{e}dei rational functions $Q_n(D,z)$ over a ring. 
\end{remark}
\subsection{The scheme}
In this section, we explicitly describe the key generation, the encryption and the decryption algorithms. 
The following steps show the key generation:
\begin{itemize}
	\item choose two prime numbers $p, q$ and compute $N=pq$;
	\item choose an integer $e$ such that $\gcd(e,\lcm(p+1)(q+1))=1$. \\
	      The pair $(N,e)$ is called the \emph{public} or \emph{encryption key};
	\item evaluate $d=e^{-1}\pmod{\lcm(p+1)(q+1)}$. \\
              The triple $(p,q,d)$ is called the \emph{secret} or \emph{decryption key}.
\end{itemize}
Now, let us consider the set
$$\mathbb{\tilde Z}_N^*=
\left\{
\cfrac{m^2+D}{m^2-D}: \forall m,D\in\mathbb Z_N^*, \pm D \ \text{quadratic non-residues modulo} \ N
\right\}.$$
Let us suppose that we want to encrypt two messages $(M_x,M_y)\in\mathbb{\tilde Z}_N^* \times \mathbb Z_N^*$. The following steps describe the encryption algorithm:
\begin{itemize}
\item compute $D=\cfrac{M_x^2-1}{M_y^2}\pmod N$, so that $(M_x,M_y)\in\pell_N^*$;
\item compute $M=\iso(M_x,M_y) = \cfrac{M_x+1}{M_y}\pmod N$;
\item compute the ciphertext $C=M^{\odot_{\para} e}\pmod N=Q_e(D,M) \pmod N$.
\end{itemize}
Once $(C,D)$ are sent to the receiver,
the decryption algorithm is described as follows:
\begin{itemize}
\item compute $C^{\odot_{\para} d}\pmod N=M$;
\item retrieve the plaintexts $(M_x,M_y)$ by means of $\iso^{-1}$, i.e. \\
      $(M_x,M_y) = \iso^{-1}(M) = \left(\cfrac{M^2+D}{M^2-D}, \cfrac{2M}{M^2-D}\right)\pmod N$. 
\end{itemize}

\subsection{Some remarks}
In the previous section, we have introduced an RSA-like scheme based on the Pell equation. Our scheme has some important differences with respect to other similar schemes.

First of all, we have considered Pell's equation $x^2-Dy^2=1$ in $\mathbb Z_N$, where $D$ is not a quadratic residue, while in \cite{padhye2006public} the author worked on the complementary case, i.e., when $D$ is a quadratic residue.\\
Moreover, we have used a particular parametrization which allowed us to define an original product $\odot_{\para_{\ZZ_P}}$ connected to R\'{e}dei rational functions. 
We have seen that in this case an analogous of the Fermat little theorem holds, i.e.,
$$m^{\odot_{\para_{\ZZ_P}} p+2}=m\pmod p,\quad \forall m\in\para_{\ZZ_p},\quad p \ \text{prime},$$
when $\mathbb K$ is a finite field. 
For this reason, the decryption key is computed modulo $\lcm(p+1)(q+1)$, 
while in RSA schemes (and also in \cite{padhye2006public}) it is computed modulo $\lcm(p-1)(q-1)$.
Furthermore, it is known that R\'{e}dei rational functions can be exploited for the construction of cryptographic systems following the Dickson scheme (see, \cite{Nob} and \cite{Mul}). 
However, in these schemes the decryption key is computed modulo $\lcm(p^2-1)(q^2-1)$, 
which is much less efficient.\\
Finally, in the previous section we have seen that we can send message pairs in $\mathbb{\tilde Z}_N^* \times \mathbb Z_N^*$. 
We can observe that also in other RSA-like schemes that exploit the Pell equation we are not able to encrypt all message pairs in $\mathbb Z_N^*\times\mathbb Z_N^*$. 
In the next example we see a message pair that can be encrypted using our scheme, 
but that can not be encrypted with the analogous scheme developed in \cite{padhye2006public}.
%
\begin{example}
Let us consider $p=11$, $q=13$, and $N=143$. We choose $e=5$ as the public exponent and consequently $d=5^{-1} \pmod{168}=101$ is the secret one. 

Let us suppose that we want to send the message pair $(M_x,M_y)=(83,135)\in\mathbb{\tilde Z}_N^*\times \mathbb Z_N^*$. We obtain $D=\cfrac{M_x^2-1}{M_x}\pmod N=54$ and $M=\iso(M_x,M_y)\pmod N=61$. Now we encrypt $M$ evaluating
$$61^{\odot_\para 5}\pmod N=38.$$
If we want to retrieve the original message we compute
$$38^{\odot_\para 101}\pmod N =61$$
and
$$\iso^{-1}(61)\pmod N=(83,135).$$
\end{example}
\begin{example}
 If we try to use the scheme II in \cite{padhye2006public} for encrypting the message pair $(M_x,M_y)=(83,135)$ we have to compute
 $$Z_1=M_xM_y\pmod N=51,\quad Z_1^{-1}\pmod N=129$$
 and to solve the system
 $$\begin{cases} X-aY=Z_1\pmod N \cr X+aY=Z_1^{-1}\pmod N \end{cases}$$
 where $Y=M_y$ and the unknowns are $X$ and $a$. 
 In this way, we get
 $$X=98,\quad a=11.$$
 We can observe that in this case $a$ is not invertible in $\mathbb Z_N$ and consequently we can not compute the value $a^{-1}$ which is necessary for the decryption. 
 However, this is a very rare situation when $N$ is large, and 
 finding a non-invertible element of $\ZZ_N$ is equivalent to factorize $N$ itself.
\end{example}
%
\section{Security of the proposed scheme}\label{sec:security}
In this section we first prove that our scheme offers the same security as RSA 
in a one-to-one communication scenario. 
Then we provide evidence that our scheme is secure against partially known plaintext attacks,
and against some attacks of which RSA is vulnerable in a broadcast scenario. 
We conclude with a comment on semantic security.
\subsection{Security reduction to RSA}\label{sec:security_reduction}
We show, with a standard reduction argument, that breaking our scheme is equivalent to breaking RSA scheme.
\begin{theorem}
The following sentences are equivalent
\begin{enumerate}
 \item \label{itm:iso} There exists a probabilistic polynomial time algorithm $A1$ such that for all $C,D \in \ZZ_N^*$,  if $C = \iso(M_x,M_y)^{\odot_\para e}$ and $D = (M_x^2-1)^2/M_y^2 \pmod N$, then $A1(C,D,e,N) = (M_x,M_y)$, where $D$ quadratic non-residue modulo $N$, $(M_x,M_y)\in \tilde{\ZZ}_N \times \ZZ_N$.
 \item \label{itm:rsa} There exists a probabilistic polynomial time algorithm $A2$ such that for all $M\in \ZZ_N^*$, if $C = M^e \pmod N$, then $A2(C,e,N)=M$.
\end{enumerate}
\end{theorem}
\begin{proof}
First assume \ref{itm:iso} and prove \ref{itm:rsa}.\\
Given $C,e,N$ we want to compute $M$ using $A1$. \\
Choose $D$ random quadratic non-residue modulo $N$. \\
Compute $(M_x,M_y)=A1(C,D,e,N)$. \\
Compute $M=\iso(M_x,M_y) = \cfrac{M_x+1}{M_y}\pmod N$.\\
Now assume \ref{itm:rsa} and prove \ref{itm:iso}.\\
Given $C,D,e,N$ we want to compute $(M_x,M_y)$ using $A2$.\\
Compute $A2(C,e,N) = M$.\\
Compute $(M_x,M_y) = \iso^{-1}(M) = 
\left(\cfrac{M^2+D}{M^2-D}, \cfrac{2M}{M^2-D}\right)\pmod N$.
\\
This completes the proof.
\end{proof}
\subsection{Partially known plaintext attack}\label{sec:partial_security}
Consider the case where the attacker is provided with the one of the two plaintext coordinates. Since $D = \frac{M_x^2 - 1}{M_y^2} \pmod N$, 
the attacker has to solve the quadratic congruence $M_x^2 - D M_y^2 - 1 = 0 \pmod N$
with respect to either $M_x$ or $M_y$.
It is well known that computing square roots modulo a composite integer $N$,
when the square root exists, 
is equivalent to factoring $N$ itself.
\subsection{Broadcast applications}\label{sec:broadcast_security}
As many other RSA-like schemes, 
such as \cite{koyama1992new}, \cite{koyama1995fast}, \cite{padhye2006public}, 
our scheme is more secure than the original RSA scheme 
when considering broadcast applications.
In such a scenario the plaintext $M$ is encrypted for $r$ receivers 
using the same public exponent $e$ and different moduli $N_i$ for each receiver, 
with $i=1,\ldots,r$. 
In \cite{hastad1986n} it is shown that,
even if the messages are public linear combinations of each other, if $r>e$
it is possible to recover $M$ by solving 
a set of $r$ congruences of polynomials of degree $e$. 
However, as shown for example in \cite{koyama1992new} and \cite{koyama1995fast}, 
this kind of attacks fails when the trapdoor function 
is not a simple monomial power as in RSA.
This allows the use of smaller exponents $e$ even in broadcast scenarios.
%
%
\subsection{Semantic security}\label{sec:semantic_security}
Our scheme can be easily transformed into a semantically secure scheme 
using standard techniques 
which introduce randomness in the process of generating the ciphertext. 
An example of these techniques is given in \cite{padhye2006public}, 
when transforming scheme II into scheme III.
\section{Efficiency comparison with RSA-like schemes}\label{sec:efficiency}
%
In this section we take into consideration the most efficient and secure
published RSA-like scheme based on Pell's equation, i.e. \cite{padhye2006public}, 
and make a comparison with our proposed scheme.

Recall that the isomorphism used in the schemes of \cite{padhye2006public} is given by
\begin{align*}
 \iso: 
 \begin{cases}
  \pell_N & \to \ZZ_N^* \\
  (x,y) & \mapsto x-Dy \pmod N  
 \end{cases}
\end{align*}
while the isomorphism used in \cite{koyama1995fast},
where $\ec = \{(x,y) \in \ZZ_N^* \colon y^2+axy = x^3 \pmod N\}$, is
\begin{align*}
 \iso: 
 \begin{cases}
  \ec & \to \ZZ_N^* \\
  (x,y) & \mapsto 1+\frac{ax}{y} \pmod N = \frac{x^3}{y^2} \pmod N
 \end{cases}
\end{align*}
Both inverses are provided in Table \ref{tab:dec_cost}. In practice the public exponent $e$ is usually chosen small and with an efficient binary representation, i.e. $e=2^{2^i}+1,i=1,2,3,4$. 
Thus the exponentiation $x^e$ performed during the encryption 
is usually much faster than $x^d$, which is done during decryption.
For this reason our comparison concerns only decryption time.

Notice also that all three schemes in \cite{padhye2006public} can only encrypt messages of size $2\log N$, and thus a fair comparison with RSA is done if considering encryption messages of the same size (rather than $\log N$-bit plaintext), which means RSA exponentiation must be applied two times in parallel, 
both during encryption and decryption.

In Table \ref{tab:dec_cost}, we indicate with M the modular multiplication (including squaring), with I the modular inverse and with E$^y_\group$ the modular exponentiation $x^y$ in the group $\group$. 
Notice that, by the observation in Remark \ref{rem:redei_exp_complexity}, both the operations E$^{d}_{\ZZ_N}$ and E$^{d}_{\para}$ can be performed in $\mathcal{O}(\log_2 N)$ with respect to addition, subtraction and multiplication.

In Table \ref{tab:dec_cost} we report the decryption costs of RSA, the first two schemes of \cite{padhye2006public}, the second scheme of \cite{koyama1995fast}, and our scheme.
\begin{table}
\begin{center}
\begin{tabular}{llll}
 & Decryption costs & Ciphertext size & $\iso^{-1}(x)$\\
\hline
RSA & 2E$^{d}_{\ZZ_N}$ & 2$\log N$ & -\\
\cite{padhye2006public}-I & 1E$^{d}_{\ZZ_N}$+3M+3I & 3$\log N$ & 
 $\left( \frac{x^{-1}+x}{2}, \frac{x^{-1}-x}{2\sqrt{D}} \right)$\\
\cite{padhye2006public}-II & 1E$^{d}_{\ZZ_N}$+2M+3I & 2$\log N$ & 
 $\left( \frac{x^{-1}+x}{2}, \frac{x^{-1}-x}{2\sqrt{D}} \right)$\\
\cite{koyama1995fast}-II & 1E$^{d}_{\ZZ_N}$+6M+2I & 2$\log N$ & 
 $\left( \frac{a^2x}{(x-1)^2}, \frac{a^3x}{(x-1)^3} \right)$\\
Our scheme & 1E$^{d}_{\para}$+3M+1I & 2$\log N$ & 
 $\left( \frac{x^2+D}{x^2-D}, \frac{2x}{x^2-D} \right)$
\end{tabular}
\end{center}
\caption{Comparison of decryption costs for a $(2\log N)$-bit plaintext.}
\label{tab:dec_cost}
\end{table}

RSA is the only scheme requiring two exponentiations. All other schemes require one exponentiation and one application of the isomorphism inverse (\cite{padhye2006public}-I requires also the application of the isomorphism). Notice that our isomorphism has the lowest number of inversions, i.e. one, and only one more multiplication than \cite{padhye2006public}-II, yielding the fastest isomorphism inverse.\\
%
%
%
%
%


\section{Conclusions}
We have presented an original RSA-like scheme based on 
an isomorphism between a conic and the set $\ZZ_N^*$ with a non-standard product.
The scheme is complementary to the fastest original RSA-like scheme based on conics,
and furthermore it has a faster isomorphism inverse operation.
We have also proved that our solution provides a two times faster decryption operation 
than the original RSA scheme,
while offering the same security in a one-to-one communication and
more security in broadcast applications.
%
%
\bibliographystyle{amsalpha}
\bibliography{critto2}
%
%
%
%
%
%
%

\end{document}